\documentclass{article}

\usepackage{arxiv}

\usepackage[utf8]{inputenc}
\usepackage[T1]{fontenc}
\usepackage{lmodern}
\usepackage{microtype}

\usepackage{
amsmath,
amssymb,
amsfonts,
amsthm,
mathtools,
bm,
mathrsfs,
bbm
}
\numberwithin{equation}{section}

\usepackage{booktabs}
\usepackage{array}
\usepackage{enumitem}
\usepackage[numbers,sort&compress]{natbib}
\usepackage[hidelinks]{hyperref}
\usepackage[nameinlink,noabbrev]{cleveref}
\usepackage{url}
\usepackage{doi}

\hypersetup{
pdftitle={
Contextual Fraction Beyond Satisfiability:
Pure NAE Completeness, Exact Orbit Packings,
and Width Barriers
},
pdfauthor={Ronald Katende},
pdfsubject={
Pure NAE noncontextuality, fractional colouring,
probability-preserving equality compilation,
gated composition, treewidth,
and extension complexity
},
pdfkeywords={
contextual fraction,
noncontextual fraction,
pure NAE,
fractional chromatic number,
orbit-capacity packing,
homomorphism packing,
gated composition,
approximation hardness,
cut polytope,
treewidth,
extension complexity
}
}

\newcommand{\D}{\mathcal{D}}
\newcommand{\supp}{\operatorname{supp}}
\newcommand{\Hom}{\operatorname{Hom}}
\newcommand{\Aut}{\operatorname{Aut}}

\newcommand{\CF}{\operatorname{CF}}
\newcommand{\NCF}{\operatorname{NCF}}

\newcommand{\NAE}{\operatorname{NAE}}

\newcommand{\CUT}{\operatorname{CUT}}
\newcommand{\STAB}{\operatorname{STAB}}
\newcommand{\xc}{\operatorname{xc}}
\newcommand{\tw}{\operatorname{tw}}

\newcommand{\id}{\operatorname{id}}
\newcommand{\1}{\mathbbm{1}}

\newcommand{\R}{\mathbb{R}}

\newcommand{\B}{\mathbf{B}}
\newcommand{\X}{\mathbf{X}}

\newcommand{\eps}{\varepsilon}

\newcommand{\Mpoly}{\mathcal{M}}
\newcommand{\Fcut}{\mathcal{F}}

\newcommand{\PureNAENC}{%
\textnormal{\textsc{Pure-NAE-NC}}%
}

\newtheorem{theorem}{Theorem}[section]
\newtheorem{lemma}[theorem]{Lemma}
\newtheorem{proposition}[theorem]{Proposition}
\newtheorem{corollary}[theorem]{Corollary}

\theoremstyle{definition}
\newtheorem{definition}[theorem]{Definition}
\newtheorem{example}[theorem]{Example}

\theoremstyle{remark}
\newtheorem{remark}[theorem]{Remark}
\newtheorem{openproblem}[theorem]{Open Problem}

\title{
Contextual Fraction Beyond Satisfiability: Pure NAE Completeness, Exact Orbit Packings, and Width Barriers
}

\author{
Ronald Katende \\
Department of Mathematics \\
Kabale University \\
Kikungiri Hill, Katuna Road \\
Kabale, Uganda \\
\texttt{rkatende@kab.ac.ug}
}

\date{}

\renewcommand{\headeright}{}
\renewcommand{\undertitle}{}
\renewcommand{\shorttitle}{
Contextual Fraction Beyond Satisfiability
}
\begin{document}

\raggedbottom
\maketitle

\maketitle

\begin{abstract}
The contextual fraction measures how much globally consistent probability can be packed beneath prescribed local event probabilities. We study the quantitative regime in which supported global assignments exist but the contextual fraction is nontrivial.

For balanced relational lifts, we identify the noncontextual fraction with a capacity-constrained packing of global homomorphisms and prove an exact orbit-capacity quotient under finite group symmetries. For the uniform Boolean relation $\NAE_3$, this becomes a minority-position congestion game on proper hypergraph two-colourings:
\[
\NCF(e_H)=\frac{1}{3\beta(H)}.
\]
We also obtain a minimax dual, a cut-polytope formulation, and explicit values $2/3$ and $1/2$.

For exact hardness, we map a graph $G$ to an always-two-colourable anchor hypergraph $A(G)$ with
\[
\NCF(e_{A(G)})=1 \Longleftrightarrow \chi_f(G)\le3.
\]
We prove quantitative stability in terms of $\chi_f(G)$ and construct a six-edge pure-$\NAE_3$ equality gadget that preserves the full noncontextual fraction under bounded-occurrence compilation. Thus deciding $\CF=0$ is NP-complete for simple $3$-uniform hypergraphs of maximum degree seven, even with a supplied proper two-colouring.

We further prove
\[
\NCF(\widehat e^{\,\vartheta})
=\vartheta+(1-\vartheta)\NCF(e).
\]
This yields NP-hardness of distinguishing $\NCF=1/2$ from $\NCF\ge2/3$ on an always-satisfiable fixed five-ary switch template, and additive hardness below $1/12$. A gated colouring construction gives hardness within $1/2-\varepsilon$ for every $\varepsilon>0$. Finally, bounded-treewidth instances admit exact compact extended formulations, while explicit families have extension complexity $2^{\Omega(\sqrt N)}$ via cut-polytope projections.
\end{abstract}


\section{Introduction}
\label{sec:introduction}

The contextual fraction was introduced as a quantitative measure of
contextuality and is computed, for finite empirical models, by a primal--dual
linear program \cite{AbramskyBarbosaMansfield2017}.  In the sheaf-theoretic
setting, supported global assignments are global sections
\cite{AbramskyBrandenburger2011}.  Their existence is closely connected with
constraint satisfaction \cite{AbramskyGottlobKolaitis2013}.  Consequently,
endpoint statements of the form
\[
\CF(e)<1
\quad\Longleftrightarrow\quad
\text{a finite constraint instance is satisfiable}
\]
inherit ordinary CSP or PCSP complexity.  The companion gain-graph paper
identifies an exact holonomy tractability island inside that support-hard
world \cite{Katende2026Gain}.

The present paper studies a different question.  Assume that supported global
assignments already exist.  How much probability mass can be placed on them
without exceeding any prescribed local event capacity?  This is not a
maximum-satisfaction objective: every global assignment used by the packing
satisfies every constraint.  The optimization concerns the simultaneous
congestion of local events across a distribution of satisfying assignments.

\paragraph{Relation to the companion gain-graph paper.}
The companion paper studies probability-weighted permutation gain graphs,
where holonomy gives an exact fixed-point formula, a query-optimal static
algorithm, and fixed-tree dynamic maintenance
\cite{Katende2026Gain}.  The present paper concerns non-permutation relational
templates and the complementary quantitative-complexity regime: orbit-capacity
ackings, pure-$\NAE_3$ completeness, approximation hardness,
bounded-treewidth formulations, and extension-complexity barriers.  The
principal theorems and proofs of the two papers are disjoint.

\subsection{Novel theorem-level contributions}

The paper is organized around six results that do not follow from the
support correspondence.

\begin{enumerate}[label=\textbf{C\arabic*.},leftmargin=2.7em]
\item \textbf{Exact orbit-capacity quotient.}
Point-transitive symmetry with several relation orbits does not merely give a
partial lower bound.  It yields an exact quotient of the contextual-fraction
LP, with variables indexed by postcomposition orbits of global homomorphisms
and capacities indexed by local relation orbits.

\item \textbf{Exact intermediate theory for Boolean $\NAE_3$.}
The quotient is identified with a minority-position congestion game.  We
prove the formula $\NCF=1/(3\beta)$, its minimax dual, a cut-face radial
representation, sharp endpoint criteria, and explicit intermediate values.

\item \textbf{Gate-free pure-$\NAE_3$ completeness and exact compilation.}
For the anchor hypergraph $A(G)$, the minority game is an explicit optimization
over the stable-set polytope, its noncontextual endpoint is equivalent to
$\chi_f(G)\le3$, and quantitative stability bounds relate its value to
$\chi_f(G)$.  A pure-$\NAE_3$ equality gadget with a unique exact local-law
kernel preserves the entire noncontextual fraction under bounded-occurrence
compilation.  This yields NP-completeness of $\CF=0$ on simple,
always-two-colourable, maximum-degree-seven hypergraphs, even when a proper
colouring is supplied.

\item \textbf{Exact gated composition.}
A free baseline branch and an active constrained branch combine affinely:
$\NCF(\widehat e^{\,\vartheta})=
\vartheta+(1-\vartheta)\NCF(e)$.  The theorem is an equality for every source
instance, not merely a completeness construction.

\item \textbf{Constant-gap hardness inside the satisfiable region.}
The $\NAE_3$ gate gives a fixed $1/6$ value gap at $\vartheta=1/2$.  A gated
three-colouring template gives $\vartheta$ versus $1$, so additive hardness
approaches $1/2$ arbitrarily closely while every relational instance and every
lifted empirical model has supported global assignments.  The gate is now
needed only for a constant additive gap, not for exact hardness.

\item \textbf{A width-sensitive polyhedral boundary.}
The homomorphism packing has an exact compact extended formulation on bounded
primal treewidth.  Conversely, an explicit family of $\NAE_3$ marginal faces,
and a face of the same gated $\NAE_3$ template used for hardness, inherit
exponential cut-polytope extension complexity.
\end{enumerate}

\subsection{What is deliberately not reproved}

The contextual-fraction LP and its resource-theoretic properties are due to
Abramsky, Barbosa, and Mansfield
\cite{AbramskyBarbosaMansfield2017}.  The contextuality--CSP bridge predates
this work \cite{AbramskyGottlobKolaitis2013}.  The algebraic theory of PCSPs,
minions, and promise reductions is developed elsewhere
\cite{BartoBulinKrokhinOprsal2021,KrokhinOprsal2022}; finite-valued CSPs and
valued-minion approaches to approximation are likewise established theories
\cite{ThapperZivny2016,BartoEtAlValuedMinions2024}.  We use only the minimum
background required to state the new packing results.

Group averaging is also classical.  PCP folding imposes algebraic invariance
on proof tables \cite{Hastad2001,AustrinBrownCohenHastad2021}, while Holant
methods exploit transformations and symmetries of local signatures
\cite{CaiLuXia2011,CaiLuXia2018}.  The novelty claimed here is not averaging
itself, but the exact residual orbit-capacity LP and the affine gated
composition law for contextual fraction.

Robust CSP and PCSP algorithms optimize the fraction of constraints satisfied
by one assignment, often through SDP rounding
\cite{BrakensiekGuruswamiSandeep2025}.  MAX NAE-SAT has its own approximation
theory \cite{BrakensiekHuangPotechinZwick2025}.  Our objective is different:
it packs a subprobability distribution supported entirely on satisfying
assignments and caps each local event separately.  Section
\ref{sec:gated-hardness} gives a worked separation on the switch template
itself.

Free-acceptance branches, completeness padding, and gap-amplification gadgets
are standard in PCP and robust-satisfaction reductions
\cite{Hastad2001,AustrinBrownCohenHastad2021,
BrakensiekGuruswamiSandeep2025,BrakensiekCiardoEtAl2026}.  We do not claim
novelty for adding an always accepting branch.  In those frameworks the
tracked quantity is normally an acceptance probability or the maximum fraction
of constraints satisfied by one assignment, and recent robust-PCSP work also
studies preservation of near-satisfaction under equality constraints.  Here one
pair of shared gate variables selects a single global branch, the two branches
occupy disjoint local-event capacities, and the tracked quantity is the maximum
mass of a distribution supported entirely on globally satisfying assignments.
The result is therefore an exact affine identity for every source packing, not
a completeness/soundness inequality.  Standard acceptance-value bookkeeping
does not by itself imply that identity.

\subsection{Relation to \texorpdfstring{$\mathsf P$ versus $\mathsf{NP}$}{P versus NP}}

The results are completeness and inapproximability theorems, not an
unconditional separation.  The pure-$\NAE_3$ theorem below gives one fixed,
bounded-arity, bounded-occurrence, always-satisfiable language on which deciding
$\CF=0$ is equivalent to an NP-complete problem.  Hence a polynomial-time
algorithm would imply $\mathsf P=\mathsf{NP}$, but the reduction supplies no
unconditional lower bound and therefore no evidence resolving the separation.
The near-$1/2$ theorem is optimal only as an additive-hardness statement up to
an arbitrarily small constant, because the constant output $1/2$ approximates
every $[0,1]$-valued quantity within error $1/2$.

\section{Balanced lifts and the exact homomorphism packing}
\label{sec:framework}

We work with finite, explicitly represented, rational probability tables.
Many-sorted notation is used because the sharp gated-colouring construction
has a two-element gate sort and a $q$-element colour sort.

\begin{definition}[Balanced realization]
Let $\B$ be a finite relational structure with sorts $s\in\mathcal S$ and
finite domains $B_s$.  A relation symbol $R$ has sort signature
$(s_1,\ldots,s_r)$ and interpretation
\[
R^{\B}\subseteq B_{s_1}\times\cdots\times B_{s_r}.
\]
A \emph{balanced realization} consists of full-support distributions
$\pi_s\in\D(B_s)$ and, for each $R$, a distribution
$\mu_R\in\D(R^{\B})$ whose $j$th marginal is $\pi_{s_j}$.
We assume $\supp(\mu_R)=R^{\B}$ unless stated otherwise.
\end{definition}

\begin{remark}[Sort discipline]
All automorphisms and homomorphisms preserve sorts.  The orbit quotient below
uses a group acting transitively on each sort separately.  The Boolean
$\NAE_3$ sections then specialize to one sort, whereas the gated-colouring
section adjoins a two-element gate sort to the unchanged colour sort.  No proof
identifies values belonging to different sorts.
\end{remark}

Let $\X$ be a finite instance.  Every variable has a sort, and every
constraint occurrence is written
$c=(R,(x_1,\ldots,x_r))$.  Isolated variables are deleted, since they do not
affect any local table.

The \emph{balanced occurrence lift} $e_{\X}$ is formed as follows.  Retain one
central observable for every variable.  For every occurrence of $x_j$ in a
constraint $c$, introduce a fresh observable $x_j^c$.  Put the table $\mu_R$
on the occurrence context $(x_1^c,\ldots,x_r^c)$ and the consistency table
\[
\mu_{=,s}(a,b):=\pi_s(a)\1[a=b]
\]
on $(x_j,x_j^c)$.  All overlaps are singleton overlaps and every table has the
prescribed singleton marginal, so the lift is compatible.

For a finite empirical model $e$, the known contextual-fraction primal is
\begin{equation}
\label{eq:cf-primal-general}
\NCF(e)=
\max\left\{
\sum_g b_g:
\sum_{g:g|_C=u}b_g\le e_C(u),\quad b_g\ge0
\right\},
\end{equation}
where $g$ ranges over global assignments and $(C,u)$ over local events
\cite{AbramskyBarbosaMansfield2017}.

\begin{proposition}[Exact homomorphism packing]
\label{prop:hom-packing}
For the balanced occurrence lift of $\X$,
\begin{align}
\NCF(e_{\X})=\max\quad
&\sum_{h\in\Hom(\X,\B)}w_h
\label{eq:hom-pack-objective}\\
\text{subject to}\quad
&\sum_{h:\,h(c)=\mathbf b}w_h\le\mu_R(\mathbf b),
&&c=(R,\mathbf x),\quad \mathbf b\in R^{\B},
\label{eq:hom-pack-capacity}\\
&w_h\ge0.
\nonumber
\end{align}
The consistency-context and singleton-capacity inequalities are redundant.
The dual is
\begin{align}
\min\quad
&\sum_{c=(R,\mathbf x)}\sum_{\mathbf b\in R^{\B}}
\mu_R(\mathbf b)y_{c,\mathbf b}
\label{eq:hom-dual-objective}\\
\text{subject to}\quad
&\sum_c y_{c,h(c)}\ge1,
&&h\in\Hom(\X,\B),
\label{eq:hom-dual-cover}\\
&y_{c,\mathbf b}\ge0.
\nonumber
\end{align}
\end{proposition}

\begin{proof}
A global column violating a consistency context or selecting a tuple outside
$R^{\B}$ meets a zero-capacity row and therefore has zero weight in
\eqref{eq:cf-primal-general}.  The remaining columns assign one value to every
central variable and preserve every relation; they are exactly the
homomorphisms $h:\X\to\B$.  Their relation-row constraints are precisely
\eqref{eq:hom-pack-capacity}.

Fix a variable $x$ of sort $s$, an incident constraint $c$, and a position $j$
where $x$ occurs.  For any $a\in B_s$,
\[
\sum_{h:h(x)=a}w_h
=\sum_{\mathbf b:\,b_j=a}
\sum_{h:h(c)=\mathbf b}w_h
\le\sum_{\mathbf b:\,b_j=a}\mu_R(\mathbf b)
=\pi_s(a).
\]
Hence all consistency and singleton capacities follow from one incident
relation context.  This proves the primal formula.  The displayed dual is the
ordinary finite LP dual.
\end{proof}

\begin{definition}[Homomorphism marginal polytope]
Let $M(h)$ be the $0$--$1$ vector recording every relation event selected by
$h$.  Define
\[
\Mpoly(\X,\B):=
\operatorname{conv}\{M(h):h\in\Hom(\X,\B)\}.
\]
Then $\NCF(e_{\X})$ is the largest $W\ge0$ for which
$Wm\le v_e$ coordinatewise for some $m\in\Mpoly(\X,\B)$.
\end{definition}

This radial-gauge viewpoint is used below.  Unlike a valued-CSP objective, it
is not a sum of local costs evaluated on one assignment.

\section{Exact orbit-capacity quotient}
\label{sec:orbit}

Assume that a finite group $G\le\Aut(\B)$ acts transitively on every sort
$B_s$.  For a relation $R$, let
\[
\Omega_R:=R^{\B}/G
\]
be the diagonal tuple-orbit set.  Choose positive orbit masses
$\theta_R\in\D(\Omega_R)$ and define
\begin{equation}
\label{eq:orbit-tables}
\pi_s(a):=\frac1{|B_s|},
\qquad
\mu_R(\mathbf b):=
\frac{\theta_R(\mathcal O)}{|\mathcal O|}
\quad(\mathbf b\in\mathcal O).
\end{equation}
Uniform measure on each tuple orbit has uniform marginal in every coordinate,
so \eqref{eq:orbit-tables} is balanced.

The group acts on $\Hom(\X,\B)$ by postcomposition.  Write
\[
\mathscr H_G(\X):=\Hom(\X,\B)/G.
\]
For $Q\in\mathscr H_G(\X)$ and a constraint $c=(R,\mathbf x)$, define
\[
\omega_Q(c):=G\,h(c)\in\Omega_R,
\]
where $h\in Q$.  This does not depend on the representative.

\begin{theorem}[Exact orbit-capacity quotient]
\label{thm:orbit-quotient}
For the orbit-balanced lift \eqref{eq:orbit-tables},
\begin{align}
\NCF(e_{\X})=\max\quad
&\sum_{Q\in\mathscr H_G(\X)}u_Q
\label{eq:orbit-objective}\\
\text{subject to}\quad
&\sum_{Q:\,\omega_Q(c)=\mathcal O}u_Q
\le\theta_R(\mathcal O),
&&c=(R,\mathbf x),\quad \mathcal O\in\Omega_R,
\label{eq:orbit-capacity}\\
&u_Q\ge0.
\nonumber
\end{align}
Thus the quotient is exact even when a relation has several tuple orbits.
\end{theorem}

\begin{proof}
Start from \cref{prop:hom-packing}.  For any feasible vector $(w_h)$, average
under postcomposition:
\[
\bar w_h:=\frac1{|G|}\sum_{g\in G}w_{g^{-1}\circ h}.
\]
Relation capacities are $G$-invariant, so the averaged vector is feasible and
has the same objective.  It is constant on every orbit
$Q\in\mathscr H_G(\X)$.  Put $u_Q:=\sum_{h\in Q}\bar w_h$.

Fix $Q$, a constraint $c$, and
$\mathcal O=\omega_Q(c)$.  The evaluation map
\[
Q\longrightarrow\mathcal O,
\qquad h\longmapsto h(c),
\]
is $G$-equivariant and surjective.  All fibers therefore have the same size.
Uniform mass $u_Q$ on $Q$ contributes exactly $u_Q/|\mathcal O|$ to every
local tuple in $\mathcal O$.  Summing over $Q$ and comparing with the capacity
$\theta_R(\mathcal O)/|\mathcal O|$ gives
\eqref{eq:orbit-capacity}.

Conversely, given feasible $(u_Q)$, distribute $u_Q$ uniformly over the
homomorphisms in $Q$.  The same equal-fiber calculation gives local load
\[
\frac1{|\mathcal O|}
\sum_{Q:\,\omega_Q(c)=\mathcal O}u_Q
\le\frac{\theta_R(\mathcal O)}{|\mathcal O|}
\]
on each tuple of $\mathcal O$.  Hence the lifted homomorphism weights are
feasible and have objective $\sum_Qu_Q$.
\end{proof}

\begin{corollary}[One-orbit endpoint]
\label{cor:one-orbit}
If $G$ is transitive on every sort and every relation, then
$\NCF(e_{\X})=1$ whenever $\X\to\B$, and $\NCF(e_{\X})=0$ otherwise.
\end{corollary}

\begin{proof}
Each $\Omega_R$ has one orbit of capacity one.  If a homomorphism orbit exists,
assign it unit mass in \cref{thm:orbit-quotient}; if none exists, the packing
has no columns.
\end{proof}

\begin{remark}[Why this is more than folding]
The averaging step is familiar from symmetry reduction, PCP folding, and
signature transformations.  The theorem identifies the exact object left
after averaging fails to collapse a relation to one orbit: a global packing
over homomorphism orbits with independent local orbit capacities.  Neither a
MAX-CSP objective nor a partition-function transformation yields
\eqref{eq:orbit-capacity} automatically.
\end{remark}

\section{Boolean \texorpdfstring{$\NAE_3$}{NAE3}: minority games and cut geometry}
\label{sec:nae}

Let
\[
B=\{0,1\},
\qquad
\NAE_3=B^3\setminus\{000,111\}.
\]
The complementation group $G=\{\id,\neg\}$ is point-transitive and has three
orbits on $\NAE_3$.  Orbit $\mathcal O_j$ contains the two tuples whose $j$th
coordinate is the unique minority bit.  The uniform $\NAE_3$ table assigns
orbit mass $1/3$ to every $\mathcal O_j$.

Throughout this section, $H=(V,E)$ is an ordered $3$-uniform hypergraph with
$E\ne\varnothing$.  The empty-edge case is noncontextual and has $\NCF=1$.
A proper two-colouring is nonconstant on every hyperedge.  We identify a
colouring with its global complement and write $\mathcal C(H)$ for the set of
complement classes.  If $e=(v_1,v_2,v_3)$ and $S\in\mathcal C(H)$, let
$m_e(S)\in\{1,2,3\}$ be the minority position.

\begin{theorem}[Exact minority-position packing]
\label{thm:nae-packing}
If $\mathcal C(H)=\varnothing$, then $\NCF(e_H)=0$.  Otherwise,
\begin{align}
\NCF(e_H)=\max\quad
&\sum_{S\in\mathcal C(H)}u_S
\label{eq:nae-objective}\\
\text{subject to}\quad
&\sum_{S:\,m_e(S)=j}u_S\le\frac13,
&&e\in E,
\quad j\in\{1,2,3\},
\label{eq:nae-capacity}\\
&u_S\ge0.
\nonumber
\end{align}
\end{theorem}

\begin{proof}
Apply \cref{thm:orbit-quotient}.  Postcomposition orbits are precisely proper
colourings modulo complementation, and the three local orbit capacities are
$1/3$.
\end{proof}

Define the \emph{minority congestion}
\begin{equation}
\label{eq:beta}
\beta(H):=
\min_{q\in\D(\mathcal C(H))}
\max_{e\in E,\,j\in[3]}
\Pr_{S\sim q}[m_e(S)=j].
\end{equation}

\begin{theorem}[Radial and minimax formulas]
\label{thm:nae-beta}
For every two-colourable nonempty $H$,
\begin{equation}
\label{eq:ncf-beta}
\NCF(e_H)=\frac1{3\beta(H)}.
\end{equation}
Moreover,
\begin{equation}
\label{eq:beta-game}
\beta(H)=
\max_{y\in\D(E\times[3])}
\min_{S\in\mathcal C(H)}
\sum_{e\in E}y_{e,m_e(S)}.
\end{equation}
Consequently,
\[
\frac13\le\beta(H)\le1,
\qquad
\frac13\le\NCF(e_H)\le1.
\]
\end{theorem}

\begin{proof}
Write any nonzero feasible vector in \eqref{eq:nae-objective} as $u=Wq$,
where $W=\sum_Su_S$ and $q$ is a distribution on $\mathcal C(H)$.  The
capacity constraints become
\[
W\max_{e,j}\Pr_q[m_e(S)=j]\le\frac13.
\]
For a fixed $q$, the largest feasible $W$ is the reciprocal of three times the
displayed maximum.  Minimizing over $q$ proves \eqref{eq:ncf-beta}.

Let $A$ be the $0$--$1$ matrix with rows $(e,j)$, columns $S$, and
$A_{(e,j),S}=\1[m_e(S)=j]$.  Then
\[
\beta(H)=
\min_{q\in\D(\mathcal C(H))}
\max_{y\in\D(E\times[3])}y^{\mathsf T}Aq.
\]
The finite minimax theorem permits interchanging min and max.  For fixed $y$,
the minimum over $q$ is attained at a pure colouring class, yielding
\eqref{eq:beta-game}.  For every fixed edge, the three minority probabilities
sum to one, giving the displayed range.
\end{proof}

\begin{corollary}[Sharp endpoints]
\label{cor:nae-endpoints}
For a two-colourable nonempty $H$:
\begin{enumerate}[label=(\roman*)]
\item $\NCF(e_H)=1$ if and only if some distribution on proper colourings
makes every minority position on every edge have probability $1/3$;
\item $\NCF(e_H)=1/3$ if and only if some event $(e,j)$ is frozen, meaning
$m_e(S)=j$ for every $S\in\mathcal C(H)$.
\end{enumerate}
\end{corollary}

\begin{proof}
Part (i) is equality $\beta=1/3$.  A frozen event gives $\beta=1$.  Conversely,
if no event is frozen, the uniform distribution on the finite set
$\mathcal C(H)$ gives every event probability strictly below one, so
$\beta<1$.
\end{proof}

\subsection{Cut-face representation}

For $S\subseteq V$, let $\delta_S(uv)=1$ when exactly one of $u,v$ lies in
$S$.  The cut polytope is
\[
\CUT(V):=\operatorname{conv}\{\delta_S:S\subseteq V\};
\]
see \cite{DezaLaurent1997} for its classical geometry.  For an ordered edge
$e=(a,b,c)$, define
\begin{align}
L_{e,1}(x)&:=\frac{x_{ab}+x_{ac}-x_{bc}}2,\nonumber\\
L_{e,2}(x)&:=\frac{x_{ab}+x_{bc}-x_{ac}}2,
\label{eq:minority-linear}\\
L_{e,3}(x)&:=\frac{x_{ac}+x_{bc}-x_{ab}}2.\nonumber
\end{align}
Set
\begin{equation}
\label{eq:proper-face}
\Fcut_H:=
\left\{
x\in\CUT(V):
x_{ab}+x_{ac}+x_{bc}=2
\text{ for every }\{a,b,c\}\in E
\right\}.
\end{equation}

\begin{theorem}[Exact cut-polytope radial formula]
\label{thm:cut-radial}
The polytope $\Fcut_H$ is the convex hull of cut vectors of proper
colourings of $H$.  Moreover,
\begin{equation}
\label{eq:beta-cut}
\beta(H)=
\min_{x\in\Fcut_H}\max_{e\in E,\,j\in[3]}L_{e,j}(x),
\end{equation}
and
\begin{equation}
\label{eq:ncf-cut}
\NCF(e_H)=
\left(
3\min_{x\in\Fcut_H}\max_{e,j}L_{e,j}(x)
\right)^{-1}.
\end{equation}
Finally,
\begin{equation}
\label{eq:two-thirds-characterization}
\NCF(e_H)=1
\quad\Longleftrightarrow\quad
\exists x\in\CUT(V)
\text{ with }
x_{uv}=\frac23
\text{ for every pair }uv\subseteq e\in E.
\end{equation}
\end{theorem}

\begin{proof}
Every cut meets a triangle in zero or two cut edges, so
$x_{ab}+x_{ac}+x_{bc}\le2$ is valid for $\CUT(V)$.  Equality selects exactly
the cut vertices nonmonochromatic on that hyperedge.  Intersecting the
corresponding supporting hyperplanes over all hyperedges proves the first
claim.

For a proper cut vector, $L_{e,j}$ is the indicator that position $j$ is the
minority on $e$.  By linearity, if $x$ is generated by a distribution on
proper cuts, then $L_{e,j}(x)$ is the corresponding minority probability.
This proves \eqref{eq:beta-cut}, and \eqref{eq:ncf-cut} follows from
\cref{thm:nae-beta}.

If all three minority probabilities equal $1/3$, solving
\eqref{eq:minority-linear} gives
$x_{ab}=x_{ac}=x_{bc}=2/3$.  Conversely, pair values $2/3$ give all three
minority probabilities $1/3$ and their triangle sum is two.  This proves
\eqref{eq:two-thirds-characterization}.
\end{proof}

\subsection{Two verified intermediate values}

The following examples are fully enumerated in
\cref{app:enumeration}; no computational assertion is left to ``inspection.''

\begin{example}[$\NCF=2/3$]
\label{ex:ncf-two-thirds}
Let
\[
V=\{1,2,3,4,5\},
\qquad
E_{2/3}=\{123,124,125,134,135,145\}.
\]
The five colouring classes are represented by
\[
00111,\quad01011,\quad01101,\quad01110,\quad01111.
\]
Weight the first four by $1/6$ and the fifth by zero.  The enumeration in
\cref{tab:h23} shows that every edge-position capacity is used at most twice,
so the primal value is $2/3$.  In the dual
\eqref{eq:hom-dual-objective}--\eqref{eq:hom-dual-cover}, charge the
first-position events of edges $123$ and $145$ by one.  Every colouring class
meets at least one charged event, while the dual objective is $2/3$.
Therefore $\NCF=2/3$.
\end{example}

\begin{example}[$\NCF=1/2$]
\label{ex:ncf-half}
Add edge $234$ to obtain $E_{1/2}=E_{2/3}\cup\{234\}$.  Exactly three
colouring classes remain:
\[
00111,\quad01011,\quad01101.
\]
Weight each by $1/6$.  The table in \cref{tab:h12} verifies all capacities.
For the dual, charge the first-position events of $125,135,145$ by $1/2$.
Every remaining colouring meets exactly two charged events, and the dual
objective is $1/2$.  Thus $\NCF=1/2$.
\end{example}

\section{Pure \texorpdfstring{$\NAE_3$}{NAE3} completeness without a gate}
\label{sec:pure-nae-hardness}

This section resolves the gate-free exact-complexity problem.  The only local
relation is Boolean $\NAE_3$, every local table is uniform on its six tuples,
and every constructed hypergraph is two-colourable.  The reduction first
identifies an exact fractional-colouring endpoint and then removes the single
high-occurrence anchor by a probability-preserving equality gadget.

\subsection{The anchor hypergraph and the stable-set polytope}

Let $G=(V,E)$ be a finite simple graph with no isolated vertices.  Choose an
orientation of every edge.  The \emph{anchor hypergraph} $A(G)$ has vertex set
$V\cup\{r\}$ and one ordered hyperedge
\[
(r,u,v)
\]
for every oriented edge $u\to v$ of $G$.  It is always properly two-colourable:
set $r=0$ and every vertex of $G$ to one.

Write $\mathcal I(G)$ for the independent sets of $G$ and
\[
\STAB(G):=\operatorname{conv}
\{\1_I:I\in\mathcal I(G)\}\subseteq\R^V
\]
for the stable-set polytope.

\begin{theorem}[Exact anchor formula]
\label{thm:anchor-formula}
For every graph $G$ without isolated vertices,
\begin{equation}
\label{eq:anchor-beta}
\beta(A(G))
=
\min_{x\in\STAB(G)}
\max_{uv\in E}
\bigl\{x_u,x_v,1-x_u-x_v\bigr\}.
\end{equation}
The value is independent of the chosen edge orientations.
\end{theorem}

\begin{proof}
Fix the representative of each colouring class with $r=0$.  For a proper
colouring $S$, let
\[
I(S):=\{v\in V:S(v)=0\}.
\]
The triple $(r,u,v)$ is nonmonochromatic exactly when $u$ and $v$ do not both
belong to $I(S)$.  Hence $S\mapsto I(S)$ is a bijection from proper colouring
classes of $A(G)$ to independent sets of $G$.

Let $q$ be a distribution on independent sets and put
$x_v=\Pr_{I\sim q}[v\in I]$.  Since $I$ is independent, the events
$u\in I$ and $v\in I$ are disjoint for every edge $uv$.  On the ordered triple
$(r,u,v)$ the three minority-position probabilities are therefore
\[
\Pr[m=1]=1-x_u-x_v,
\qquad
\Pr[m=2]=x_v,
\qquad
\Pr[m=3]=x_u.
\]
Conversely every $x\in\STAB(G)$ is the marginal vector of a distribution on
independent sets.  Substitution into the definition of $\beta$ proves
\eqref{eq:anchor-beta}.  Reversing an edge merely exchanges the last two
entries of the displayed set.
\end{proof}

We use the standard fractional-colouring linear program
\begin{equation}
\label{eq:fractional-chromatic}
\chi_f(G)=
\min\left\{
\sum_{I\in\mathcal I(G)}\lambda_I:
\lambda_I\ge0,
\ \sum_{I\ni v}\lambda_I\ge1\ \ (v\in V)
\right\};
\end{equation}
see \cite[Chapter~3]{ScheinermanUllman2011}.

\begin{lemma}[Uniform stable-set point]
\label{lem:uniform-stab}
For every graph $G$ and every real $t\ge1$,
\[
\chi_f(G)\le t
\quad\Longleftrightarrow\quad
t^{-1}\1\in\STAB(G).
\]
\end{lemma}

\begin{proof}
Suppose first that $t^{-1}\1=\sum_Iq_I\1_I$ for a distribution $q$ on
independent sets.  Then $\lambda_I=tq_I$ is feasible in
\eqref{eq:fractional-chromatic} with objective $t$.

Conversely, let $(\lambda_I)$ be feasible with total mass at most $t$.  Add
weight to the empty independent set so that the total is exactly $t$, and set
$q_I=\lambda_I/t$.  Its marginal vector $y$ satisfies $y_v\ge1/t$ for every
$v$.  The stable-set polytope is down-monotone: from a random independent set
with marginal vector $y$, independently retain an included vertex $v$ with
conditional probability $(1/t)/y_v$.  The retained set is independent and has
marginal vector $t^{-1}\1$.  Thus $t^{-1}\1\in\STAB(G)$.
\end{proof}

\begin{theorem}[Fractional-colouring endpoint]
\label{thm:anchor-fractional-colouring}
For every graph $G$ without isolated vertices,
\begin{align}
\beta(A(G))=\frac13
&\quad\Longleftrightarrow\quad
\chi_f(G)\le3,
\label{eq:beta-frac}\\
\NCF(e_{A(G)})=1
&\quad\Longleftrightarrow\quad
\chi_f(G)\le3.
\label{eq:ncf-frac}
\end{align}
\end{theorem}

\begin{proof}
For every edge $uv$, the three numbers in \eqref{eq:anchor-beta} sum to one,
so their maximum is at least $1/3$.  If the optimum equals $1/3$, then on every
edge all three numbers equal $1/3$.  Hence $x_u=x_v=1/3$ for every edge; since
there are no isolated vertices, $x=(1/3)\1$.  Conversely that vector makes all
three terms equal $1/3$.  By \cref{lem:uniform-stab},
$(1/3)\1\in\STAB(G)$ exactly when $\chi_f(G)\le3$.  This proves
\eqref{eq:beta-frac}, and \eqref{eq:ncf-frac} follows from
\cref{thm:nae-beta}.
\end{proof}

Thus noncontextuality of the uniform pure-$\NAE_3$ model already detects a
classical fractional-colouring threshold, while every anchor instance has an
obvious global section.

\begin{proposition}[Fixed-threshold fractional colouring]
\label{prop:fractional-three-npcomplete}
For every fixed rational $t>2$, the problem of deciding whether
$\chi_f(G)\le t$ is NP-complete.  In particular, deciding whether
$\chi_f(G)\le3$ is NP-complete.
\end{proposition}

\begin{proof}
NP-hardness is the fixed-threshold theorem stated in
\cite[Theorem~3.9.2]{ScheinermanUllman2011}.  We record membership in NP because
it is needed explicitly below.  The fractional-colouring linear program has
one covering constraint for each vertex.  If its optimum is at most the fixed
rational $t$, then it has a basic feasible solution supported on at most
$|V(G)|$ independent sets.  The corresponding weights solve a rational linear
system with a $0$--$1$ coefficient matrix.  Cramer's rule and Hadamard's bound
therefore give polynomial encoding length.  A certificate lists those
independent sets and rational weights; independence, all vertex-covering
inequalities, and total weight at most $t$ are checkable in polynomial time.
\end{proof}

\begin{remark}[Hardness source]
\label{rem:fractional-hardness-source}
The reduction below uses only the fixed unweighted threshold $t=3$, exactly as
stated in \cite[Theorem~3.9.2]{ScheinermanUllman2011}.  It does not use the
weighted optimization theorem of Gr\"otschel, Lov\'asz, and Schrijver.
Approximation hardness for fractional chromatic number is developed in
\cite{LundYannakakis1994}.
\end{remark}

\subsection{Quantitative stability of the anchor construction}

\begin{theorem}[Two-sided anchor stability]
\label{thm:anchor-stability}
Let $G$ have no isolated vertices and put $t=\chi_f(G)\ge3$.  Then
\begin{equation}
\label{eq:beta-two-sided}
\frac{t-1}{2t}
\le \beta(A(G))
\le 1-\frac{2}{t}.
\end{equation}
Consequently,
\begin{equation}
\label{eq:ncf-two-sided}
\frac{t}{3(t-2)}
\le \NCF(e_{A(G)})
\le \frac{2t}{3(t-1)}.
\end{equation}
At $t=3$, both bounds are equalities: $\beta(A(G))=1/3$ and
$\NCF(e_{A(G)})=1$.
\end{theorem}

\begin{proof}
Let $b=\beta(A(G))$ and choose $x\in\STAB(G)$ attaining
\eqref{eq:anchor-beta}.  For every edge $uv$,
\[
x_u\le b,\qquad x_v\le b,\qquad x_u+x_v\ge1-b.
\]
Every vertex has a neighbour, and hence
\[
x_v\ge1-b-x_u\ge1-2b
\qquad(v\in V(G)).
\]
If $b<1/2$, down-monotonicity of $\STAB(G)$ gives
$(1-2b)\1\in\STAB(G)$.  By \cref{lem:uniform-stab},
$t\le(1-2b)^{-1}$, or equivalently
$b\ge(t-1)/(2t)$.  If $b\ge1/2$, the same inequality is automatic because
$(t-1)/(2t)<1/2$.

For the upper bound, \cref{lem:uniform-stab} gives
$t^{-1}\1\in\STAB(G)$.  Substitution into \eqref{eq:anchor-beta} yields
\[
b\le\max\left\{\frac1t,1-\frac2t\right\}
=1-\frac2t
\]
for $t\ge3$.  Applying \cref{thm:nae-beta} and inverting the positive bounds
proves \eqref{eq:ncf-two-sided}.
\end{proof}

\begin{corollary}[Quantitative endpoint stability]
\label{cor:anchor-endpoint-stability}
If $0\le\delta<1/3$ and
$\NCF(e_{A(G)})\ge1-\delta$, then
\begin{equation}
\label{eq:chi-from-ncf}
\chi_f(G)\le\frac{3(1-\delta)}{1-3\delta}.
\end{equation}
Equivalently, if $\chi_f(G)\ge3+\eta$, then
\begin{equation}
\label{eq:ncf-from-chi-gap}
\NCF(e_{A(G)})
\le1-\frac{\eta}{3(2+\eta)}.
\end{equation}
\end{corollary}

\begin{proof}
The hypothesis and $\NCF=1/(3\beta)$ imply
$\beta\le[3(1-\delta)]^{-1}$.  The lower-bound argument in
\cref{thm:anchor-stability} gives
$\chi_f(G)\le(1-2\beta)^{-1}$ and hence \eqref{eq:chi-from-ncf}.
Substituting $t=3+\eta$ into the upper bound in
\eqref{eq:ncf-two-sided} gives \eqref{eq:ncf-from-chi-gap}.
\end{proof}

\subsection{A probability-preserving pure-NAE equality gadget}

The anchor $r$ has unbounded occurrence.  We now compile every repeated
variable occurrence through a constant-size gadget without introducing any
relation other than $\NAE_3$.

For terminals $x,y$ and fresh auxiliary variables $a,b,c$, let
$\mathcal Q(x,y;a,b,c)$ contain the six hyperedges
\begin{equation}
\label{eq:equality-gadget}
\begin{split}
&(x,a,b),\ (x,a,c),\ (x,b,c),\\
&(y,a,b),\ (y,a,c),\ (y,b,c).
\end{split}
\end{equation}

\begin{lemma}[Support and unique exact local-law kernel]
\label{lem:equality-gadget}
The gadget $\mathcal Q(x,y;a,b,c)$ has a proper two-colouring extension if and
only if $x=y$.  Fix a random supported extension for which the common terminal
bit $T:=x=y$ is uniform.  Then all six gadget edges are uniformly distributed
on $\NAE_3$ if and only if, conditional on $T=t$, exactly one of $a,b,c$ is
chosen uniformly to equal $t$ and the other two equal $1-t$.  In particular,
the probability-preserving conditional extension kernel is unique.
\end{lemma}

\begin{proof}
Fix $x=0$.  The first three constraints say that no two of $a,b,c$ can both be
zero, so at most one auxiliary bit equals zero.  Fixing $x=1$ analogously says
that at most one auxiliary bit equals one.  If $x\ne y$, both requirements hold
simultaneously for three bits, which is impossible.  If $x=y=t$, the supported
auxiliary assignments are exactly
\[
o_t:=(1-t,1-t,1-t)
\quad\text{and}\quad
e_{t,a}:=(t,1-t,1-t),\quad
e_{t,b}:=(1-t,t,1-t),\quad
e_{t,c}:=(1-t,1-t,t).
\]
Thus equality of the terminals is necessary and sufficient for support
satisfiability.

Let $\alpha_t,p_{t,a},p_{t,b},p_{t,c}$ be the conditional probabilities of
$o_t,e_{t,a},e_{t,b},e_{t,c}$ given $T=t$.  On the edge $(x,a,b)$, the three
possible $\NAE_3$ tuples with first coordinate $t$ occur with conditional
probabilities
\[
p_{t,a},\qquad p_{t,b},\qquad \alpha_t+p_{t,c}.
\]
Because $T$ is uniform, uniformity of this edge on all six $\NAE_3$ tuples is
equivalent to
\[
p_{t,a}=p_{t,b}=\alpha_t+p_{t,c}=\frac13.
\]
Applying the same argument to $(x,a,c)$ gives
\[
p_{t,a}=p_{t,c}=\alpha_t+p_{t,b}=\frac13.
\]
Hence
\[
p_{t,a}=p_{t,b}=p_{t,c}=\frac13,
\qquad \alpha_t=0.
\]
This holds for both $t=0$ and $t=1$, proving uniqueness.  Conversely, the
stated kernel gives, on any gadget edge, the three $\NAE_3$ tuples with first
coordinate $t$ with conditional probability $1/3$ each; averaging over the
uniform terminal gives all six tuples with probability $1/6$.  The same
calculation applies to the three edges containing $y$, since $y=x$ on the
support.
\end{proof}

Given a hypergraph $H$, form its \emph{bounded-occurrence compilation}
$\widetilde H$ as follows.  Create one copy of a vertex for each of its
hyperedge occurrences and use those copies in the corresponding source
hyperedges.  For every original vertex, connect its occurrence copies in a
path, replacing every path equality by a fresh gadget
\eqref{eq:equality-gadget}.  Every edge of $\widetilde H$ carries the uniform
$\NAE_3$ table.

\begin{theorem}[Exact value-preserving compilation]
\label{prop:compilation}
After isolated vertices are deleted, every pure-$\NAE_3$ hypergraph $H$
satisfies
\begin{equation}
\label{eq:exact-compilation}
\boxed{\NCF(e_{\widetilde H})=\NCF(e_H).}
\end{equation}
The compiled hypergraph is simple and $3$-uniform.  Every occurrence copy has
degree at most seven and every auxiliary vertex has degree four.  If a proper
two-colouring of $H$ is supplied, a proper two-colouring of
$\widetilde H$ is computable in linear time.
\end{theorem}

\begin{proof}
We work with the unquotiented homomorphism-packing formulation, whose variables
are the proper two-colourings themselves and whose local tuple capacities are
$1/6$.

First let $(w_h)$ be a feasible packing for $H$.  For a vertex $v$ and
$t\in\{0,1\}$ put
\[
M_v(t):=\sum_{h:\,h(v)=t}w_h.
\]
Choose an edge incident with $v$ and sum the three tuple-capacity inequalities
whose $v$-coordinate is $t$.  Each has capacity $1/6$, so
\begin{equation}
\label{eq:terminal-half-bound}
M_v(t)\le\frac12.
\end{equation}
For each colouring $h$, assign every occurrence copy the value of its original
vertex and extend each equality gadget independently by the kernel from
\cref{lem:equality-gadget}: conditional on the common terminal value $t$,
choose uniformly which one of the three auxiliary vertices equals $t$.
Distribute the mass $w_h$ over the resulting compiled colourings according to
the product of these kernels.

Every source edge retains exactly its old tuple load.  On a gadget edge and
for a fixed terminal value $t$, each of the three supported tuples with that
terminal value has conditional probability $1/3$.  Its total load is therefore
$M_v(t)/3\le1/6$ by \eqref{eq:terminal-half-bound}.  Thus the extended packing
is feasible and has the same total mass, proving
$\NCF(e_{\widetilde H})\ge\NCF(e_H)$.

Conversely, every supported colouring of $\widetilde H$ is constant along each
path of occurrence copies, because every equality gadget forces its two
terminals to agree.  Collapse each path to the corresponding original vertex
and aggregate the weights of compiled colourings with the same collapsed
colouring.  Source-edge tuple loads are unchanged and hence remain at most
$1/6$.  The aggregated vector is a feasible packing for $H$ with the same
objective value.  Therefore
$\NCF(e_{\widetilde H})\le\NCF(e_H)$, establishing
\eqref{eq:exact-compilation}.

An occurrence copy lies in one source edge and in at most two path gadgets;
each terminal appears in three edges of a gadget, so its degree is at most
$1+2\cdot3=7$.  Each auxiliary appears in four gadget edges.  Simplicity and
$3$-uniformity are immediate because every gadget uses fresh auxiliaries.  A
proper source colouring extends by choosing, in every gadget, one auxiliary
equal to the common terminal and the other two opposite to it.
\end{proof}

\begin{theorem}[Conditional constant-gap transfer]
\label{thm:conditional-pure-gap}
Fix $\eta>0$.  Suppose it is NP-hard to distinguish graphs satisfying
\[
\chi_f(G)\le3
\qquad\text{from graphs satisfying}\qquad
\chi_f(G)\ge3+\eta.
\]
Then it is NP-hard, on simple $3$-uniform hypergraphs of maximum degree seven
supplied with a proper two-colouring, to distinguish
\[
\NCF(e_H)=1
\qquad\text{from}\qquad
\NCF(e_H)\le1-\frac{\eta}{3(2+\eta)}.
\]
\end{theorem}

\begin{proof}
Map $G$ to the anchor $A(G)$.  The YES case follows from
\cref{thm:anchor-fractional-colouring}; the NO-case bound is
\eqref{eq:ncf-from-chi-gap}.  Apply the exact compiler
\cref{prop:compilation}.  The colouring with the anchor equal to zero and all
graph vertices equal to one extends through every equality gadget and supplies
the promised proper two-colouring.
\end{proof}

\begin{remark}[Normalization bottleneck]
\label{rem:normalization-bottleneck}
General approximation hardness for fractional chromatic number is known
\cite{LundYannakakis1994}, but it does not automatically give the normalized
promise $3$ versus $3+\eta$.  The theorem isolates exactly the additional gap
statement needed to obtain gate-free constant-gap hardness for pure
$\NAE_3$.
\end{remark}

\subsection{A bounded-occurrence NP-completeness theorem}

Let \PureNAENC{} denote the following promise problem.
The input is a simple $3$-uniform hypergraph $H$ of maximum vertex degree at
most seven, together with a proper two-colouring of $H$.  Every edge carries
the uniform distribution on $\NAE_3$.  The question is whether $\CF(e_H)=0$.

\begin{theorem}[Pure NAE noncontextuality is NP-complete]
\label{thm:pure-nae-npcomplete}
\PureNAENC{} is NP-complete.  Consequently, exact
evaluation of $\CF$ is NP-hard on this promised family, and deciding
$\CF(e_H)>0$ is coNP-complete.
\end{theorem}

\begin{proof}
For NP-hardness, start from the fixed-threshold fractional-colouring problem
$\chi_f(G)\le3$.  Delete isolated vertices, which does not affect this
inequality unless the graph is edgeless, a trivial case.  Construct $A(G)$ and
then its bounded-occurrence compilation $\widetilde{A(G)}$.  By
\cref{thm:anchor-fractional-colouring,prop:compilation},
\[
\chi_f(G)\le3
\quad\Longleftrightarrow\quad
\NCF(e_{\widetilde{A(G)}})=1
\quad\Longleftrightarrow\quad
\CF(e_{\widetilde{A(G)}})=0.
\]
The construction is linear in the number of graph-edge incidences.  The
colouring with all anchor copies zero and all graph-vertex copies one extends
through every equality gadget, so the promised colouring is produced by the
reduction.

For membership in NP, let $m$ be the number of hyperedges.  Noncontextuality is
feasibility of a rational linear system whose columns are proper colourings and
whose rows prescribe the six event probabilities on each edge.  If feasible,
a basic feasible solution uses at most $6m+1$ colourings.  Because the
constraint matrix has entries in $\{0,1\}$ and the right-hand side has entries
in $\{1,1/6\}$, Cramer's rule and Hadamard's bound give polynomial encoding
length for all basic weights.  A certificate listing these colourings and
weights is therefore polynomial in the input size and can be checked directly.
The complement and exact-evaluation claims follow.
\end{proof}

\begin{corollary}[Exact relation to $\mathsf P$ versus $\mathsf{NP}$]
\label{cor:pnp-equivalence}
The following are equivalent:
\begin{enumerate}[label=(\roman*)]
\item $\mathsf P=\mathsf{NP}$;
\item \PureNAENC{} is decidable in polynomial time;
\item there is a polynomial-time exact evaluator for $\CF$ on the same
promised family.
\end{enumerate}
This is a completeness equivalence, not an unconditional separation result.
\end{corollary}

\begin{proof}
The implication (iii)$\Rightarrow$(ii) is immediate and
(ii)$\Rightarrow$(i) follows from \cref{thm:pure-nae-npcomplete}.  If
$\mathsf P=\mathsf{NP}$, rational LP optimization over a polynomially
verifiable, polynomial-support certificate family is computable in polynomial
time by the standard decision-to-optimization reduction, giving (iii).
\end{proof}

\begin{remark}[What has and has not been removed]
The switch relation is absent, the local language is exactly uniform
$\NAE_3$, support satisfiability is promised, and occurrence degree is bounded.
The construction still uses equality-compilation gadgets.  Exact value
preservation does not by itself supply a constant additive gap.  Removing those remaining features,
or proving a constant gap on the pure family, is a strictly stronger problem.
\end{remark}

\section{Exact gated composition}
\label{sec:gated-composition}

We now introduce the operation that converts any balanced source packing into
an always-satisfiable quantitative instance without losing its value.

Let the source template $\B$ have balanced realization
$(\pi_s,\mu_R)$.  Introduce a gate sort
\[
M=\{0,1\}
\]
with uniform marginal.  For every source relation
$R\subseteq B_{s_1}\times\cdots\times B_{s_r}$, define a gated relation
$\widehat R$ on
$M\times M\times B_{s_1}\times\cdots\times B_{s_r}$ by
\begin{align}
\widehat R_0
&:=\{(a,a,\mathbf b):a\in M,
\ \mathbf b\in B_{s_1}\times\cdots\times B_{s_r}\},
\label{eq:gate-baseline-support}\\
\widehat R_1
&:=\{(a,1-a,\mathbf b):a\in M,
\ \mathbf b\in R^{\B}\},
\label{eq:gate-active-support}\\
\widehat R&:=\widehat R_0\cup\widehat R_1.
\nonumber
\end{align}
Define probability tables
\begin{align}
\nu_{R,0}(a,a,\mathbf b)
&:=\frac12\prod_{j=1}^r\pi_{s_j}(b_j),
\label{eq:gate-baseline-table}\\
\nu_{R,1}(a,1-a,\mathbf b)
&:=\frac12\mu_R(\mathbf b),
\label{eq:gate-active-table}\\
\widehat\mu_R^{\,\vartheta}
&:=\vartheta\nu_{R,0}+(1-\vartheta)\nu_{R,1},
\qquad 0<\vartheta<1.
\label{eq:gated-table}
\end{align}
Both branch tables have uniform gate marginals and source-coordinate marginals
$\pi_s$, so the gated template is balanced.

Let $\X$ be a nonempty \emph{simple} source instance: every constraint scope
contains pairwise distinct variables and every variable occurs in a
constraint.  Construct $\widehat\X$ by adding two shared gate variables
$r,z\in M$ and replacing every source constraint
$c=(R,(x_1,\ldots,x_r))$ by
\[
\widehat c=
\widehat R(r,z,x_1,\ldots,x_r).
\]
Setting $z=r$ satisfies all gated relations independently of the source
variables.

\begin{theorem}[Exact affine gated-sum identity]
\label{thm:gated-affine}
For every simple nonempty source instance $\X$,
\begin{equation}
\label{eq:gated-affine}
\boxed{
	\NCF(e_{\widehat\X}^{\,\vartheta})
	=
	\vartheta+(1-\vartheta)\NCF(e_{\X}).
}
\end{equation}
Every gated instance has supported global assignments, and in fact
$\NCF(e_{\widehat\X}^{\,\vartheta})\ge\vartheta$.
\end{theorem}

\begin{proof}
Every homomorphism of $\widehat\X$ belongs to exactly one of two global
branches.  In the baseline branch $z=r$ and the source variables are
unrestricted.  In the active branch $z\ne r$ and the source-variable
restriction is a homomorphism $h:\X\to\B$.

Let a feasible gated packing have baseline mass $W_0$ and active mass $W_1$.
Fix one gated constraint.  Summing its capacity inequalities over all tuples
in $\widehat R_0$ gives
\[
W_0\le
\widehat\mu_R^{\,\vartheta}(\widehat R_0)
=\vartheta.
\]
For each source homomorphism $h$, aggregate the active gated weights over the
two gate orientations and rescale:
\[
\bar w_h:=
\frac1{1-\vartheta}
\sum_{a\in\{0,1\}}
w_{(r=a,\,z=1-a,\,h)}.
\]
For a source constraint $c=(R,\mathbf x)$ and tuple
$\mathbf b\in R^{\B}$,
\begin{align*}
\sum_{h:h(c)=\mathbf b}\bar w_h
&=\frac1{1-\vartheta}
\sum_{a\in\{0,1\}}
\sum_{h:h(c)=\mathbf b}
w_{(a,1-a,h)}\\
&\le\frac1{1-\vartheta}
\sum_{a\in\{0,1\}}
\frac{1-\vartheta}{2}\mu_R(\mathbf b)
=\mu_R(\mathbf b).
\end{align*}
Thus $(\bar w_h)$ is feasible for the source packing, so
\[
W_1\le(1-\vartheta)\NCF(e_{\X}).
\]
This proves the upper bound in \eqref{eq:gated-affine}.

For the reverse inequality, first construct a baseline global distribution:
choose $r=z$ uniformly and choose every source variable independently from its
sort marginal $\pi_s$.  Simplicity of the scopes makes the local source tuple
law the product distribution in \eqref{eq:gate-baseline-table}.  Scale this
global distribution by mass $\vartheta$.

Next take an optimal source packing $(w_h)$.  For every source homomorphism
$h$ and every gate orientation $a\in\{0,1\}$, assign gated weight
\[
\frac{1-\vartheta}{2}w_h
\]
to $(r=a,z=1-a,h)$.  At an active relation event the load is at most
$\frac{1-\vartheta}{2}\mu_R(\mathbf b)$, exactly its capacity.  At a source
singleton event the baseline contributes $\vartheta\pi_s(b)$ and the active
packing contributes at most $(1-\vartheta)\pi_s(b)$; gate singleton loads are
at most $1/2$.  Hence the two subpackings can be added, giving total mass
\[
\vartheta+(1-\vartheta)\NCF(e_{\X}).
\]
\end{proof}

\begin{remark}[Exact composition, not gap bookkeeping]
A standard free-acceptance gadget proves only that one branch is always
available.  The new statement is the exact identity
\eqref{eq:gated-affine}: no packing mass can leak between the branches, and
all active mass is precisely a scaled source packing.  This permits exact
values, sharp hardness gaps, and polyhedral inheritance.
\end{remark}

\section{Hardness within always-satisfiable families}
\label{sec:gated-hardness}

\subsection{The Boolean \texorpdfstring{$\NAE_3$}{NAE3} switch: exact value and improved gap}

For the Boolean source, \eqref{eq:gate-baseline-support}--
\eqref{eq:gate-active-support} give the five-ary relation
\[
\mathsf S(r,z,x,y,w)
\quad\Longleftrightarrow\quad
(z=r)
\ \text{or}\ 
(z\ne r\text{ and }\NAE_3(x,y,w)).
\]
Its baseline branch has $16$ tuples and its active branch has $12$ tuples;
with uniform Boolean marginals, \eqref{eq:gated-table} is exactly
\[
\vartheta U_{\mathsf S_0}
+(1-\vartheta)U_{\mathsf S_1}.
\]
Given a $3$-uniform hypergraph $H$, use shared variables $r,z$ and one switch
constraint per hyperedge.

\begin{corollary}[Exact switch formula]
\label{cor:nae-switch-exact}
For every nonempty $3$-uniform hypergraph $H$,
\begin{equation}
\label{eq:nae-switch-exact}
\NCF(e_{\widehat H}^{\,\vartheta})
=
\begin{cases}
	\vartheta,
	& H\text{ is not two-colourable},\\[1mm]
	\displaystyle
	\vartheta+\frac{1-\vartheta}{3\beta(H)},
	& H\text{ is two-colourable}.
\end{cases}
\end{equation}
Every switch instance is satisfiable.
\end{corollary}

\begin{proof}
Combine \cref{thm:gated-affine,thm:nae-beta}.  If $H$ is not two-colourable,
the source packing has no columns and value zero.
\end{proof}

Two-colourability of $3$-uniform hypergraphs, equivalently monotone
$\NAE_3$-satisfiability, is NP-complete by Schaefer's dichotomy
\cite{Schaefer1978}.  Since every colourable nonempty $H$ satisfies
$\NCF(e_H)\ge1/3$, \cref{cor:nae-switch-exact} gives the following stronger
version of the earlier switch bound.

\begin{theorem}[Improved always-satisfiable $\NAE_3$ hardness]
\label{thm:nae-switch-hardness}
For every fixed rational $\vartheta\in(0,1)$, it is NP-hard on an
always-satisfiable fixed-template family to distinguish
\[
\NCF(e)=\vartheta
\qquad\text{from}\qquad
\NCF(e)\ge
\vartheta+\frac{1-\vartheta}{3}.
\]
Therefore additive approximation of either $\NCF$ or $\CF$ within error
strictly smaller than $(1-\vartheta)/6$ is NP-hard.  At
$\vartheta=1/2$, the gap is
\[
\frac12
\qquad\text{versus}\qquad
\NCF(e)\ge\frac23,
\]
and the forbidden additive error is $1/12$.
\end{theorem}

\begin{proof}
The source NO case has value zero; every source YES case has value at least
$1/3$ by \cref{thm:nae-beta}.  Apply the affine identity
\eqref{eq:gated-affine}.  An additive error smaller than half the value gap
distinguishes the two cases.
\end{proof}

\begin{remark}[Worked separation from MAX-NAE and robust satisfaction]
\label{rem:max-nae-separation}
For every switch instance, setting $z=r$ satisfies every switch constraint, so
the maximum fraction of constraints satisfiable by one assignment is exactly
one, independently of $H$.  The contextual packing can nevertheless be
strictly below one.  For the hypergraph in \cref{ex:ncf-half} and
$\vartheta=1/2$, the exact affine law gives
\[
\NCF(e_{\widehat H}^{1/2})
=\frac12+\frac12\cdot\frac12
=\frac34,
\]
although the MAX-switch value is one.  Thus the objective is not MAX-NAE-SAT,
not robust satisfiability, and not a relabelling of the fraction of satisfied
constraints studied in \cite{BrakensiekGuruswamiSandeep2025,
BrakensiekHuangPotechinZwick2025}.
\end{remark}

\subsection{Near-maximal additive hardness from gated colouring}

Let the source colour sort be $C=[q]$ and let
\[
\mathrm{NEQ}_q:=\{(a,b)\in C^2:a\ne b\}
\]
carry the uniform table $1/[q(q-1)]$.  For a graph $G$, denote the balanced
colouring lift by $e_G^{(q)}$.

\begin{lemma}[Colouring source is zero--one]
\label{lem:colour-zero-one}
For every graph $G$,
\[
\NCF(e_G^{(q)})=
\begin{cases}
1,&G\text{ is }q\text{-colourable},\\
0,&G\text{ is not }q\text{-colourable}.
\end{cases}
\]
\end{lemma}

\begin{proof}
If no proper colouring exists, the homomorphism packing has no columns.  If
$c$ is a proper colouring, choose a uniformly random permutation
$\sigma\in S_q$ and use the global colouring $\sigma\circ c$.  Every vertex is
uniform on $[q]$, and every ordered unequal pair on every edge occurs with
probability $1/[q(q-1)]$.  Hence the full empirical model is a marginal of a
global distribution and has $\NCF=1$.
\end{proof}

Apply the gate construction to $\mathrm{NEQ}_q$.  The resulting fixed
many-sorted relation is
\[
\widehat{\mathrm{NEQ}}_q(r,z,a,b)
\quad\Longleftrightarrow\quad
(z=r)
\ \text{or}\
(z\ne r\text{ and }a\ne b).
\]
The baseline table is uniform on
$\{z=r\}\times C^2$ and the active table is uniform on
$\{z\ne r\}\times\mathrm{NEQ}_q$.

\begin{theorem}[Near-maximal additive hardness]
\label{thm:near-half-hardness}
Fix $q\ge3$ for which $q$-colourability is NP-complete, and fix a rational
$\vartheta\in(0,1)$.  On an always-satisfiable family over the fixed gated
colouring template, it is NP-hard to distinguish
\[
\NCF(e)=\vartheta
\qquad\text{from}\qquad
\NCF(e)=1.
\]
Consequently, additive approximation of $\NCF$ or $\CF$ within any error
strictly smaller than
\[
\frac{1-\vartheta}{2}
\]
is NP-hard.
\end{theorem}

\begin{proof}
By \cref{lem:colour-zero-one,thm:gated-affine}, a non-$q$-colourable graph
gives value $\vartheta$, while a $q$-colourable graph gives
$\vartheta+(1-\vartheta)=1$.  Every gated instance is satisfiable through the
baseline branch.  For $q=3$, graph colouring is NP-complete
\cite{Karp1972,GareyJohnson1979}.
\end{proof}

\begin{corollary}[Optimal ceiling approached]
\label{cor:additive-ceiling}
For every $\eps\in(0,1/2)$, there exists a fixed finite balanced template for
which additive approximation within $1/2-\eps$ is NP-hard on an
always-satisfiable family.
\end{corollary}

\begin{proof}
Choose a fixed positive rational $\vartheta<2\eps$ and apply
\cref{thm:near-half-hardness}; then
$(1-\vartheta)/2>1/2-\eps$.
\end{proof}

\begin{proposition}[Explicit accuracy-to-template encoding]
\label{prop:epsilon-encoding}
For $q=3$ and $\eps\in(0,1/2)$, put
\[
M:=\left\lfloor\frac1{2\eps}\right\rfloor+1,
\qquad \vartheta:=\frac1M.
\]
Then $\vartheta<2\eps$, the gated relation has fixed arity four and fixed
support size thirty, and its probability table is described by the two atomic
weights
\[
\frac{1}{18M}
\quad\text{on each baseline tuple},
\qquad
\frac{M-1}{12M}
\quad\text{on each active tuple}.
\]
Thus the relational support and domain sizes are independent of $\eps$, while
the complete rational table requires $O(\log(1/\eps))$ bits.
\end{proposition}

\begin{proof}
The choice of $M$ gives $M>(2\eps)^{-1}$ and hence $1/M<2\eps$.
For $q=3$, the baseline branch $z=r$ contains $2q^2=18$ tuples and the active
branch $z\ne r$, $a\ne b$ contains $2q(q-1)=12$ tuples.  Multiplying their
uniform laws by $\vartheta$ and $1-\vartheta$ gives the displayed atomic
weights.  The numerator and denominator bit lengths are $O(\log M)$, and
$M=O(1/\eps)$.
\end{proof}

\begin{remark}[Sharpness and scope]
No $[0,1]$-valued invariant can have an additive inapproximability threshold
strictly above $1/2$, because the constant estimate $1/2$ always has error at
most $1/2$.  Thus \cref{cor:additive-ceiling} is optimal up to an arbitrarily
small constant.  The template depends on $\eps$ through the fixed rational
branch weight.  Obtaining the full ceiling with one positive-baseline fixed
template is impossible by this two-branch gap alone and remains a separate
amplification question.
\end{remark}

\section{Bounded treewidth and extension-complexity barriers}
\label{sec:polyhedral-boundary}

The preceding hardness results concern unrestricted instances.  We now locate
a precise structural tractability boundary and then show that it cannot be
replaced by a width-independent compact LP.

\subsection{Exact formulation on bounded primal treewidth}

Let the primal graph of $\X$ connect two variables when they occur in a common
constraint.  Let $(T,(B_t)_{t\in V(T)})$ be a tree decomposition of width $t$,
and assign every constraint to one bag containing its scope.  Put
$b:=\max_s|B_s|$.

\begin{theorem}[Treewidth-exact extended formulation]
\label{thm:treewidth}
For every fixed finite balanced template, $\NCF(e_{\X})$ is computable by an
exact LP with
\[
O\bigl(|T|b^{t+1}\bigr)
\]
variables and constraints, up to a template-dependent constant.  Given the
decomposition, the running time is
\[
b^{O(t)}\operatorname{poly}(|\X|).
\]
\end{theorem}

\begin{proof}
For every bag $B_t$ and every locally satisfying assignment
$s:B_t\to\bigcup_s B_s$, introduce a nonnegative variable $\eta_{t,s}$.  For
adjacent bags $t,u$, impose equality of their separator marginals:
\[
\sum_{s:s|_{B_t\cap B_u}=a}\eta_{t,s}
=
\sum_{s':s'|_{B_t\cap B_u}=a}\eta_{u,s'}.
\]
At a root bag impose $\sum_s\eta_{r,s}=W$ and maximize $W$.  If a constraint
$c=(R,\mathbf x)$ is assigned to bag $t$, impose for every
$\mathbf b\in R^{\B}$
\[
\sum_{s:s(\mathbf x)=\mathbf b}\eta_{t,s}
\le\mu_R(\mathbf b).
\]
Singleton capacities are redundant by the marginal argument in
\cref{prop:hom-packing}.

Every global homomorphism packing induces feasible bag submarginals.  Conversely,
consistent nonnegative bag tables of common total mass glue on a tree:
normalize when $W>0$, sample the root bag, and recursively sample each child
conditioned on its separator.  The running-intersection property gives a
global assignment; local satisfaction ensures that it is a homomorphism.
Rescaling by $W$ yields a global packing with the prescribed bag marginals.
Thus the formulation is exact.  The state count is
$O(b^{t+1})$ per bag.  This is the capacity-packing refinement of the known
bounded-treewidth CSP extended-formulation method
\cite{KolmanKoutecky2015}.
\end{proof}

\begin{corollary}[Gating preserves bounded width]
\label{cor:gating-treewidth}
If $\widehat\X$ is obtained from $\X$ by the shared two-variable gate, then
\[
\tw(\X)\le\tw(\widehat\X)\le\tw(\X)+2.
\]
Hence the exact gated packing remains fixed-parameter tractable in source
primal treewidth.
\end{corollary}

\begin{proof}
The source primal graph is an induced subgraph of the gated primal graph.  For
the upper bound, add the two shared gate variables to every bag of a tree
decomposition of the source graph.
\end{proof}

\subsection{Cut-polytope projection}

For $n\ge2$, let $U=[n]$.  Introduce an auxiliary vertex $z_{ij}$ for every
pair $1\le i<j\le n$, and define the $3$-uniform hypergraph
\[
H_n:
\qquad
E(H_n)=\{\{i,j,z_{ij}\}:1\le i<j\le n\}.
\]
Let $\Fcut_{H_n}$ be the proper-cut face from \eqref{eq:proper-face}.

\begin{theorem}[Cut projection and extension lower bound]
\label{thm:cut-xc}
Projection of $\Fcut_{H_n}$ onto the coordinates
$\{x_{ij}:i,j\in U\}$ is exactly $\CUT(K_n)$.  Consequently,
\[
\xc(\Fcut_{H_n})
\ge\xc(\CUT(K_n))
=2^{\Omega(n)}.
\]
Since $H_n$ has $N=\Theta(n^2)$ vertices, the lower bound is
$2^{\Omega(\sqrt N)}$ in $N$.
\end{theorem}

\begin{proof}
The projection of any cut on the enlarged vertex set is a cut on $U$, so it
lies in $\CUT(K_n)$.  Conversely, fix a cut of $U$.  For each pair $i,j$, if
$i$ and $j$ have the same colour, colour $z_{ij}$ oppositely; if they differ,
choose either colour.  Every triple $\{i,j,z_{ij}\}$ is then
nonmonochromatic.  Thus every vertex of $\CUT(K_n)$ lifts to a vertex of
$\Fcut_{H_n}$, and convexity gives equality of the projection.

Extension complexity cannot increase under projection from the source
polytope.  The exponential lower bound for the complete cut polytope is due
to Fiorini, Massar, Pokutta, Tiwary, and de Wolf
\cite{FioriniEtAl2015}.
\end{proof}

The primal graph of $H_n$ contains $K_n$ on the original vertex set $U$, so
$\tw(H_n)\ge n-1$.  Therefore \cref{thm:treewidth,thm:cut-xc} are consistent:
the lower-bound family has unbounded width.

\subsection{The obstruction persists for the same gated \texorpdfstring{$\NAE_3$}{NAE3} template}

Let $\Mpoly_{\mathrm{sw}}(H)$ be the homomorphism marginal polytope of the
$\NAE_3$ switch instance built from $H$.  Include singleton coordinates for
the two shared gate variables.

\begin{theorem}[Active-face inheritance]
\label{thm:gated-xc}
The face of $\Mpoly_{\mathrm{sw}}(H_n)$ defined by
\[
r=0,
\qquad z=1
\]
projects onto $\Fcut_{H_n}$.  Hence
\[
\xc\bigl(\Mpoly_{\mathrm{sw}}(H_n)\bigr)
\ge2^{\Omega(n)}
=2^{\Omega(\sqrt N)}.
\]
Thus the fixed gated $\NAE_3$ template used in
\cref{thm:nae-switch-hardness} also has no width-independent
polynomial-size exact marginal formulation.
\end{theorem}

\begin{proof}
On the face $r=0,z=1$, every switch constraint is in its active branch and
reduces to the corresponding $\NAE_3$ relation.  Its global homomorphisms are
therefore exactly the proper two-colourings of $H_n$, with the gate coordinates
fixed.  Projecting their event vectors to pair-disagreement coordinates gives
$\Fcut_{H_n}$ by \cref{thm:cut-radial}.  A projection of a face cannot have
larger extension complexity than the original polytope, so
\cref{thm:cut-xc} gives the bound.
\end{proof}

\begin{remark}[What extension complexity does not prove]
The theorem rules out a universal compact exact LP for the full homomorphism
marginal polytope.  It does not by itself prove hardness of every scalar
objective over that polytope.  Scalar inapproximability is supplied separately
by \cref{thm:nae-switch-hardness,thm:near-half-hardness}.  The significance is
that the algorithmic and geometric obstructions now apply to the same gated
$\NAE_3$ framework rather than to unrelated templates.
\end{remark}

\section{Scope, remaining frontier, and conclusion}
\label{sec:scope}

The paper closes the main objections to an endpoint-only formulation.
Quantitative hardness is proved where all instances have supported global
assignments; the switch objective is separated explicitly from MAX-NAE and
robust satisfaction; the NAE examples are fully enumerated; the bounded-width
and extension-complexity results are reconciled; and the latter is transferred
to the same gated template used for hardness.

The results nevertheless have sharp limits.

\begin{openproblem}[Constant-gap pure $\NAE_3$ hardness]
Does there exist a universal $\delta>0$ for which it is NP-hard, on simple
bounded-degree $3$-uniform hypergraphs supplied with a proper two-colouring, to
distinguish
\[
\NCF(e_H)=1
\qquad\text{from}\qquad
\NCF(e_H)\le1-\delta?
\]
\Cref{thm:pure-nae-npcomplete} settles exact gate-free hardness, while
\cref{thm:conditional-pure-gap} identifies a sufficient normalized fractional-colouring
gap; the unconditional constant-gap strengthening remains open.
\end{openproblem}

\begin{openproblem}[Algebraic classification of orbit-capacity packings]
Characterize fixed templates for which the quotient LP in
\cref{thm:orbit-quotient} is solvable by a uniform polynomial-time LP, SDP,
affine, or other algebraic relaxation.  Existing polymorphism, fractional
polymorphism, and valued-minion theories
\cite{BartoBulinKrokhinOprsal2021,ThapperZivny2016,
BartoEtAlValuedMinions2024} do not immediately classify the radial
local-capacity gauge studied here.
\end{openproblem}

\begin{openproblem}[Single-template ceiling amplification]
Can one fixed positive-baseline template yield NP-hard additive approximation
for every constant error below $1/2$, rather than using a template whose fixed
branch weight depends on the requested error?  Any such theorem must amplify
the value gap without allowing the baseline mass to vanish and thereby
reintroducing support hardness.
\end{openproblem}

\begin{openproblem}[Approximation of the minority game]
Develop approximation algorithms and hardness thresholds for $\beta(H)$.  The
cut-face formulation suggests semidefinite relaxations, but the objective is a
worst local minority-position probability, not total cut weight or the
fraction of satisfied clauses.
\end{openproblem}

The main conceptual outcome is an exact separation of mechanisms.  Full
relation transitivity gives a zero--one endpoint.  Multiple relation orbits
leave an exact capacity quotient.  Gating composes that quantitative value
affinely with a positive always-satisfiable branch.  Bounded treewidth permits
exact compact formulations, while unbounded width supports cut-polytope
extension barriers.  The exact gate-free problem is now settled by
\cref{thm:pure-nae-npcomplete}.  The next decisive step is quantitative: a
constant-gap tractability--hardness theory for the pure minority game that
requires neither a switch relation nor a vanishing-precision decision.

\section*{Code availability}
No empirical datasets were generated or analysed. The code is used only as a reproducibility check and not as a substitute for any proof.

\appendix

\section{Complete verification of the two five-vertex examples}
\label{app:enumeration}

For a colouring representative and an ordered hyperedge, the table entry is
the minority position.  Representatives are chosen with vertex $1$ coloured
zero.

\begin{table}[ht]
\centering
\caption{All colouring classes for $E_{2/3}$ and their minority positions.}
\label{tab:h23}
\renewcommand{\arraystretch}{1.12}
\begin{tabular}{@{}c|cccccc|c@{}}
\toprule
Colouring & $123$ & $124$ & $125$ & $134$ & $135$ & $145$ & primal weight\\
\midrule
$00111$ & 3&3&3&1&1&1 & $1/6$\\
$01011$ & 2&1&1&3&3&1 & $1/6$\\
$01101$ & 1&2&1&2&1&3 & $1/6$\\
$01110$ & 1&1&2&1&2&2 & $1/6$\\
$01111$ & 1&1&1&1&1&1 & $0$\\
\midrule
Positive-row counts
& $(2,1,1)$ & $(2,1,1)$ & $(2,1,1)$
& $(2,1,1)$ & $(2,1,1)$ & $(2,1,1)$ & $2/3$\\
\bottomrule
\end{tabular}
\end{table}

The count triple $(n_1,n_2,n_3)$ records how many positive-weight rows use
each minority position.  Since every positive row has weight $1/6$ and every
count is at most two, all capacities $1/3$ are respected.  The dual charges
$y_{123,1}=y_{145,1}=1$ cover all five rows: every row has minority position
one in at least one of those two edges.  The dual objective is
$(1/3)(1+1)=2/3$.

\begin{table}[ht]
\centering
\caption{All colouring classes for $E_{1/2}=E_{2/3}\cup\{234\}$.}
\label{tab:h12}
\renewcommand{\arraystretch}{1.12}
\begin{tabular}{@{}c|ccccccc|c@{}}
\toprule
Colouring & $123$ & $124$ & $125$ & $134$ & $135$ & $145$ & $234$ & weight\\
\midrule
$00111$ & 3&3&3&1&1&1&1 & $1/6$\\
$01011$ & 2&1&1&3&3&1&2 & $1/6$\\
$01101$ & 1&2&1&2&1&3&3 & $1/6$\\
\midrule
Counts
& $(1,1,1)$ & $(1,1,1)$ & $(2,0,1)$
& $(1,1,1)$ & $(2,0,1)$ & $(2,0,1)$
& $(1,1,1)$ & $1/2$\\
\bottomrule
\end{tabular}
\end{table}

Again every count is at most two.  For the dual, set
\[
y_{125,1}=y_{135,1}=y_{145,1}=\frac12
\]
and all other variables to zero.  The three rows meet respectively the charged
edge pairs
\[
\{135,145\},\qquad
\{125,145\},\qquad
\{125,135\}.
\]
Each row therefore receives total charge one, and the dual objective is
$3\cdot(1/3)\cdot(1/2)=1/2$.

\section{A worked bounded-occurrence compilation}
\label{app:k3-compilation}

We spell out the compilation of the anchor hypergraph for $G=K_3$.  Write the
vertices of $K_3$ as $1,2,3$.  The source anchor hypergraph has edges
\[
(r,1,2),\qquad (r,1,3),\qquad (r,2,3).
\]
After occurrence splitting, use the nine terminal copies
\[
r_{12},r_{13},r_{23},
\quad 1_{12},1_{13},
\quad 2_{12},2_{23},
\quad 3_{13},3_{23}.
\]
The three source edges become
\[
(r_{12},1_{12},2_{12}),\qquad
(r_{13},1_{13},3_{13}),\qquad
(r_{23},2_{23},3_{23}).
\]
Insert equality gadgets along the five path links
\[
r_{12}{=}r_{13},\quad r_{13}{=}r_{23},\quad
1_{12}{=}1_{13},\quad 2_{12}{=}2_{23},\quad
3_{13}{=}3_{23}.
\]
Each link is replaced by a fresh copy of the six-edge gadget
$\mathcal Q$.  The compiled hypergraph therefore has $24$ vertices and
$3+5\cdot6=33$ hyperedges.  Its degree accounting is exact:

\begin{table}[ht]
\centering
\caption{Degree accounting for the compiled $A(K_3)$.}
\label{tab:k3-degree}
\begin{tabular}{@{}lrrrr@{}}
\toprule
Vertex class & Number & Source edges & Incident gadgets & Degree each\\
\midrule
Middle anchor copy $r_{13}$ & 1 & 1 & 2 & 7\\
Endpoint anchor copies $r_{12},r_{23}$ & 2 & 1 & 1 & 4\\
Graph-vertex occurrence copies & 6 & 1 & 1 & 4\\
Gadget auxiliaries & 15 & 0 & -- & 4\\
\bottomrule
\end{tabular}
\end{table}

The degree sum is
\[
1\cdot7+2\cdot4+6\cdot4+15\cdot4=99=3\cdot33,
\]
which independently checks the edge count.  A promised proper colouring is
obtained by assigning all anchor copies the bit $0$ and all graph-vertex copies
the bit $1$.  In each anchor-copy equality gadget choose one auxiliary equal
to $0$ and the other two equal to $1$; in each graph-copy gadget choose one
auxiliary equal to $1$ and the other two equal to $0$.  Every source edge then
has pattern $(0,1,1)$ and every gadget edge is nonmonochromatic.  This example
realizes the maximum degree seven at the middle anchor copy.

%

\section{Direct dual of the minority packing}
\label{app:nae-dual}

Dualizing \eqref{eq:nae-objective}--\eqref{eq:nae-capacity} gives
\begin{align}
\min\quad
&\frac13\sum_{e\in E}\sum_{j=1}^3y_{e,j}
\label{eq:nae-dual-objective}\\
\text{subject to}\quad
&\sum_{e\in E}y_{e,m_e(S)}\ge1,
&&S\in\mathcal C(H),
\label{eq:nae-dual-constraint}\\
&y_{e,j}\ge0.
\nonumber
\end{align}
It is a fractional cover of all proper-colouring classes by local
minority-position events.  Normalizing a nonzero dual vector by its total mass
produces the maximizer in \eqref{eq:beta-game}; conversely, scaling an optimal
game strategy until its minimum payoff is one gives an optimal dual cover.

\bibliographystyle{unsrtnat}
\bibliography{refs_paper_2}

\begin{thebibliography}{23}
\providecommand{\natexlab}[1]{#1}
\providecommand{\url}[1]{\texttt{#1}}
\expandafter\ifx\csname urlstyle\endcsname\relax
  \providecommand{\doi}[1]{doi: #1}\else
  \providecommand{\doi}{doi: \begingroup \urlstyle{rm}\Url}\fi

\bibitem[Abramsky et~al.(2017)Abramsky, Barbosa, and
  Mansfield]{AbramskyBarbosaMansfield2017}
Samson Abramsky, Rui~Soares Barbosa, and Shane Mansfield.
\newblock Contextual fraction as a measure of contextuality.
\newblock \emph{Physical Review Letters}, 119:\penalty0 050504, 2017.
\newblock \doi{10.1103/PhysRevLett.119.050504}.

\bibitem[Abramsky and Brandenburger(2011)]{AbramskyBrandenburger2011}
Samson Abramsky and Adam Brandenburger.
\newblock The sheaf-theoretic structure of non-locality and contextuality.
\newblock \emph{New Journal of Physics}, 13:\penalty0 113036, 2011.
\newblock \doi{10.1088/1367-2630/13/11/113036}.

\bibitem[Abramsky et~al.(2013)Abramsky, Gottlob, and
  Kolaitis]{AbramskyGottlobKolaitis2013}
Samson Abramsky, Georg Gottlob, and Phokion~G. Kolaitis.
\newblock Robust constraint satisfaction and local hidden variables in quantum
  mechanics.
\newblock In \emph{Proceedings of the Twenty-Third International Joint
  Conference on Artificial Intelligence}, pages 440--446, 2013.

\bibitem[Katende(2026)]{Katende2026Gain}
Ronald Katende.
\newblock Contextual fraction on permutation gain graphs: Exact algorithms,
  query lower bounds, and dynamic maintenance, 2026.
\newblock arXiv:2607.16037 [cs.DS].

\bibitem[Barto et~al.(2021)Barto, Bul{'i}n, Krokhin, and
  Opr{\v{s}}al]{BartoBulinKrokhinOprsal2021}
Libor Barto, Jakub Bul{'i}n, Andrei Krokhin, and Jakub Opr{\v{s}}al.
\newblock Algebraic approach to promise constraint satisfaction.
\newblock \emph{Journal of the ACM}, 68\penalty0 (4):\penalty0 28:1--28:66,
  2021.
\newblock \doi{10.1145/3457606}.

\bibitem[Krokhin and Opr{\v{s}}al(2022)]{KrokhinOprsal2022}
Andrei Krokhin and Jakub Opr{\v{s}}al.
\newblock An invitation to the promise constraint satisfaction problem.
\newblock \emph{ACM SIGLOG News}, 9\penalty0 (3):\penalty0 30--59, 2022.
\newblock \doi{10.1145/3559736.3559740}.

\bibitem[Thapper and {\v{Z}}ivn{'y}(2016)]{ThapperZivny2016}
Johan Thapper and Stanislav {\v{Z}}ivn{'y}.
\newblock The complexity of finite-valued {CSP}s.
\newblock \emph{Journal of the ACM}, 63\penalty0 (4):\penalty0 37:1--37:33,
  2016.
\newblock \doi{10.1145/2974019}.

\bibitem[Barto et~al.(2024)Barto, Butti, Kazda, Viola, and
  {\v{Z}}ivn{'y}]{BartoEtAlValuedMinions2024}
Libor Barto, Silvia Butti, Alexandr Kazda, Caterina Viola, and Stanislav
  {\v{Z}}ivn{'y}.
\newblock Algebraic approach to approximation.
\newblock In \emph{Proceedings of the 39th Annual ACM/IEEE Symposium on Logic
  in Computer Science}, pages 10:1--10:14, 2024.
\newblock \doi{10.1145/3661814.3662076}.

\bibitem[H{\aa}stad(2001)]{Hastad2001}
Johan H{\aa}stad.
\newblock Some optimal inapproximability results.
\newblock \emph{Journal of the ACM}, 48\penalty0 (4):\penalty0 798--859, 2001.
\newblock \doi{10.1145/502090.502098}.

\bibitem[Austrin et~al.(2021)Austrin, Brown-Cohen, and
  H{\aa}stad]{AustrinBrownCohenHastad2021}
Per Austrin, Jonah Brown-Cohen, and Johan H{\aa}stad.
\newblock Optimal inapproximability with universal factor graphs.
\newblock In \emph{Proceedings of the 2021 ACM--SIAM Symposium on Discrete
  Algorithms}, pages 434--453, 2021.
\newblock \doi{10.1137/1.9781611976465.27}.

\bibitem[Cai et~al.(2011)Cai, Lu, and Xia]{CaiLuXia2011}
Jin-Yi Cai, Pinyan Lu, and Mingji Xia.
\newblock Computational complexity of {Holant} problems.
\newblock \emph{SIAM Journal on Computing}, 40\penalty0 (4):\penalty0
  1101--1132, 2011.
\newblock \doi{10.1137/100814585}.

\bibitem[Cai et~al.(2018)Cai, Lu, and Xia]{CaiLuXia2018}
Jin-Yi Cai, Pinyan Lu, and Mingji Xia.
\newblock Dichotomy for real {$\operatorname{Holant}^{c}$} problems.
\newblock In \emph{Proceedings of the 29th Annual ACM--SIAM Symposium on
  Discrete Algorithms}, pages 1802--1821, 2018.
\newblock \doi{10.1137/1.9781611975031.118}.

\bibitem[Brakensiek et~al.(2025{\natexlab{a}})Brakensiek, Guruswami, and
  Sandeep]{BrakensiekGuruswamiSandeep2025}
Joshua Brakensiek, Venkatesan Guruswami, and Sai Sandeep.
\newblock {SDP}s and robust satisfiability of promise {CSP}.
\newblock \emph{Discrete Analysis}, 2025:\penalty0 14, 2025{\natexlab{a}}.
\newblock \doi{10.19086/da.143808}.

\bibitem[Brakensiek et~al.(2025{\natexlab{b}})Brakensiek, Huang, Potechin, and
  Zwick]{BrakensiekHuangPotechinZwick2025}
Joshua Brakensiek, Neng Huang, Aaron Potechin, and Uri Zwick.
\newblock On the mysteries of {MAX NAE-SAT}.
\newblock \emph{SIAM Journal on Discrete Mathematics}, 39\penalty0
  (1):\penalty0 267--313, 2025{\natexlab{b}}.
\newblock \doi{10.1137/23M1591578}.

\bibitem[Brakensiek et~al.(2026)Brakensiek, Ciardo, Guruswami, Potechin, and
  {\v{Z}}ivn{'y}]{BrakensiekCiardoEtAl2026}
Joshua Brakensiek, Lorenzo Ciardo, Venkatesan Guruswami, Aaron Potechin, and
  Stanislav {\v{Z}}ivn{'y}.
\newblock New algorithms and hardness results for robust satisfiability of
  promise {CSP}s.
\newblock In \emph{Proceedings of the 2026 Annual ACM--SIAM Symposium on
  Discrete Algorithms}, pages 3965--3977, 2026.
\newblock \doi{10.1137/1.9781611978971.145}.

\bibitem[Deza and Laurent(1997)]{DezaLaurent1997}
Michel~Marie Deza and Monique Laurent.
\newblock \emph{Geometry of Cuts and Metrics}, volume~15 of \emph{Algorithms
  and Combinatorics}.
\newblock Springer, Berlin, 1997.
\newblock \doi{10.1007/978-3-642-04295-9}.

\bibitem[Scheinerman and Ullman(2011)]{ScheinermanUllman2011}
Edward~R. Scheinerman and Daniel~H. Ullman.
\newblock \emph{Fractional Graph Theory: A Rational Approach to the Theory of
  Graphs}.
\newblock Dover Publications, Mineola, NY, 2011.
\newblock Unabridged republication of the 1997 Wiley edition.

\bibitem[Lund and Yannakakis(1994)]{LundYannakakis1994}
Carsten Lund and Mihalis Yannakakis.
\newblock On the hardness of approximating minimization problems.
\newblock \emph{Journal of the ACM}, 41\penalty0 (5):\penalty0 960--981, 1994.
\newblock \doi{10.1145/185675.306789}.

\bibitem[Schaefer(1978)]{Schaefer1978}
Thomas~J. Schaefer.
\newblock The complexity of satisfiability problems.
\newblock In \emph{Proceedings of the Tenth Annual ACM Symposium on Theory of
  Computing}, pages 216--226, 1978.
\newblock \doi{10.1145/800133.804350}.

\bibitem[Karp(1972)]{Karp1972}
Richard~M. Karp.
\newblock Reducibility among combinatorial problems.
\newblock In Raymond~E. Miller and James~W. Thatcher, editors, \emph{Complexity
  of Computer Computations}, pages 85--103. Plenum Press, New York, 1972.

\bibitem[Garey and Johnson(1979)]{GareyJohnson1979}
Michael~R. Garey and David~S. Johnson.
\newblock \emph{Computers and Intractability: A Guide to the Theory of
  {NP}-Completeness}.
\newblock W. H. Freeman, San Francisco, CA, 1979.

\bibitem[Kolman and Kouteck{'y}(2015)]{KolmanKoutecky2015}
Petr Kolman and Martin Kouteck{'y}.
\newblock Extended formulation for {CSP} that is compact for instances of
  bounded treewidth.
\newblock \emph{The Electronic Journal of Combinatorics}, 22\penalty0
  (4):\penalty0 P4.30, 2015.
\newblock \doi{10.37236/5474}.

\bibitem[Fiorini et~al.(2015)Fiorini, Massar, Pokutta, Tiwary, and {de
  Wolf}]{FioriniEtAl2015}
Samuel Fiorini, Serge Massar, Sebastian Pokutta, Hans~Raj Tiwary, and Ronald
  {de Wolf}.
\newblock Exponential lower bounds for polytopes in combinatorial optimization.
\newblock \emph{Journal of the ACM}, 62\penalty0 (2):\penalty0 17:1--17:23,
  2015.
\newblock \doi{10.1145/2716307}.

\end{thebibliography}
\end{document}